\documentclass[%
 reprint,
superscriptaddress,
 amsmath,amssymb,
 aps,
]{revtex4-2}

\usepackage{amsmath}
\usepackage {xcolor}
\usepackage{comment}
\usepackage{bbm}
\usepackage{mathrsfs}
\usepackage{comment}
\usepackage{amssymb}
\usepackage{graphicx}
\usepackage{subfigure}
\usepackage[colorlinks,
            linkcolor=blue,
            anchorcolor=blue,
            citecolor=blue]{hyperref}

\newtheorem{proposition}{Proposition}

\allowdisplaybreaks[4]

\def\ra{\rangle}
\def\la{\langle}

\usepackage{graphicx}
\usepackage{dcolumn}
\usepackage{bm}

\begin{document}

\title{Quantum average correlation based on average coherence}

\author{Xiaoyu Ma}

\affiliation{College of Science, National University of Defense Technology, Changsha 410073, China}

\author{Qing-Hua Zhang}
\email[]{qhzhang@csust.edu.cn}
\affiliation{School of Mathematics and Statistics, Changsha University of Science and Technology, Changsha 410114, China}
\affiliation{Hunan Provincial Key Laboratory of Mathematical Modeling and Analysis in Engineering, Changsha University of Science and Technology, Changsha 410114, China}

\author{Cong Xu}
\affiliation{School of Mathematical Sciences, Capital Normal University,
Beijing 100048, China}

\thanks{}

\begin{abstract}
This paper studies the quantification and structural properties of quantum average correlation based on average coherence. Motivated by two mathematically equivalent approaches to define average coherence: one by averaging over complete sets of mutually unbiased bases, and the other by integrating over all orthogonal bases under the Haar measure, we define an average correlation for bipartite systems as the difference between global and local skew information. This correlation measure is shown to satisfy essential properties including non negativity, contractivity under local quantum channels, and local unitary invariance. We further prove the equivalence between the average correlation defined via mutually unbiased bases and that defined via unitary groups. Finally, we derive a complementarity relation that connects wave-particle duality with the average correlation between a system and its environment. 
\end{abstract}

\maketitle



\section{Introduction}

Quantum coherence originates from the superposition principle of quantum mechanics, serving as the core of many quantum phenomena and a key feature distinguishing the quantum world from the classical one.
The quantification of quantum coherence, first established in Ref. \cite{PhysRevLett.113.140401}, laid the foundation for its treatment as a quantum resource. Subsequent work introduced a variety of coherence measures \cite{Rana2016,PhysRevLett.115.020403, PhysRevLett.113.170401} and explored its manipulation via distillation and conversion protocols \cite{chitambar2016assisted}. A comprehensive resource-theoretic framework for coherence is reviewed in Ref.~\cite{streltsov2017colloquium}. Within the framework of resource theory for coherence, a variety of coherence quantifiers have been successively proposed. These include the $l_p$-norm of coherence \cite{PhysRevLett.113.140401}, coherence measures based on quantum skew information \cite{PhysRevA.95.042337,sun2021quantifying}, fidelity \cite{Shao2015}, trace distance \cite{Rana2016,Yu2016,Wang2016}, geometric coherence \cite{zhang2017estimation}, and relative entropy of coherence \cite{Zhao2018,PhysRevA.93.032136,shao2017quantum,Chitambar2016}, among others \cite{PhysRevLett.113.140401}. Quantum coherence serves as a resource that makes many physical tasks possible \cite{Korzekwa_2016,PhysRevLett.129.130602,PhysRevLett.125.180603,PhysRevA.99.022340,PhysRevA.97.062342,PhysRevLett.119.150405,PhysRevLett.116.150502,PhysRevLett.116.120404,PhysRevA.98.032324}.

More recently, research has focused on the fundamental relations between coherence and other nonclassical features, such as quantum correlations and complementarity \cite{wdxh-nwsw,PhysRevA.110.042413,wwjf-lh44,fan2023average,Sun2022QuantifyingCoherence,PhysRevA.85.032117,fan2025average,Ma2016,sun2017quantum,PhysRevA.111.052451,PhysRevA.109.052439,PhysRevA.110.042413,PhysRevLett.116.160406,Zhang_2024,jin2021maximum}. In particular, Luo et al. proposed a measure for correlations in terms of the skew information \cite{PhysRevA.85.032117}. Luo and Sun proposed the average and maximal coherence based on Wigner-Yanase skew information to quantity the quantum correlation \cite{LUO20192869,sun2017quantum}. Streltsov et al. established a fundamental trade-off between the $l_1$-norm coherence measure and linear entropy with respect to a fixed reference basis \cite{PhysRevLett.115.020403}. Subsequent research has generalized this relation to the relative entropy of coherence and alternative mixedness measures such as the von Neumann entropy \cite{PhysRevA.92.042101,sun2022complementary,CHE2023106794,HU20181,PhysRevA.93.032136,PhysRevA.93.062111}.

According to the fact that the Wigner-Yanase skew information of the global state is greater than that of local state, it is naturally to capture the correlation by the difference between these two skew information. Given the basis-dependent nature of quantum correlation measures, a natural question concerns whether quantum correlations based on the Wigner-Yanase skew information can be characterized in an intrinsic, basis-independent way that depends only on the quantum state. Fortuitously, quantities such as the minimal coherence and maximal coherence allow for basis-independent formulations, achievable by optimizing over all bases \cite{PhysRevA.92.042101,LUO20192869}. Motivated by this, we aim to construct a correlation measure based on the average coherence, which provides an intrinsic characterization of a quantum state.

In this work, we develop a quantifier for quantum correlations rooted in the notion of average coherence. The paper is structured as follows. Section \ref{sec2} introduces two complementary formulations of average coherence, defined respectively via mutually unbiased bases and unitary integration over all orthogonal bases. Building on these definitions, Section \ref{sec3} constructs several correlation measures that satisfy key properties such as non-negativity, contractivity, and local unitary invariance, and establishes their equivalence under appropriate normalizations. Section \ref{sec3.2} derives a complementarity relation linking wave-particle duality with the average correlation between a system and its environment. Finally, Section \ref{sec4} summarizes the results and outlines potential extensions.

\section{Average coherence} 	\label{sec2}
Let $\rho$ be a quantum state (density operator) on a $d$-dimensional system Hilbert space and $\Pi=\{|i\rangle\langle i|: i=1,2, \cdots, d\}$ be a von Neumann measurement, $i.e.$, $\{|i\rangle: i=1,2, \cdots, d\}$ constitutes an orthonormal basis. The measure of coherence based on the Wigner-Yanase skew information is 
$$
C(\rho |\Pi)=\sum_{i=1}^d I(\rho,|i\rangle\langle i|),
$$
where  $I(\rho, O)=-\frac{1}{2} \mathbf{tr}[\sqrt{\rho}, O]^2$, $O$ is an observable, and $[\cdot, \cdot]$ denotes commutator between two operators \cite{Wigner1963INFORMATION,PhysRevA.95.042337,PhysRevA.96.022130,PhysRevA.98.012113}. The coherence measure $C(\rho |\Pi)$ satisfies several key properties that define a physically meaningful coherence quantifier, including convexity, monotonicity under incoherent operations \cite{PhysRevA.96.022130,PhysRevA.98.012113}.

The quantification of average coherence can be approached via two distinct methodologies, each corresponding to a different sampling strategy over measurement bases. The first approach involves averaging over all mutually unbiased bases (MUBs), a set of bases in which any two bases are mutually unbiased, that is, measurements in one basis give completely random outcomes with respect to the other. Recall that two orthonormal bases $\Pi_1=\left\{\left|b_{1 i}\right\rangle: i=1,2, \cdots, d\right\}$ and $\Pi_2=\left\{\left|b_{2 i}\right\rangle: i=1,2, \cdots, d\right\}$ of a \(d\)-dimensional Hilbert space are mutually unbiased if
$\left|\left\langle b_{1 j} | b_{2 k}\right\rangle\right|=1 / \sqrt{d}, \quad$ for all $j, k$.
When the dimension $d$ is a prime power ($i.e.$, $d=p^k$ for a prime number $p$ and a positive integer $k$ ), there exists a complete set of $d+1$ MUBs $\Pi_t=\left\{\left| b_{t i}\right\rangle: j=1,2, \cdots, d\right\}, t=1,2, \cdots, d+1$. Define the average coherence of $\rho$ with respect to a complete set of $\Pi_t$ \cite{LUO20192869,fan2023average}
$$
C_{{mub}}(\rho)=\frac{1}{d+1} \sum_{t=1}^{d+1} C\left(\rho | \Pi_t\right).
$$
The well-defined average coherence is actually independent of the choice of complete sets of MUBs. Averaging coherence over such a set captures a kind of intrinsic quantum randomness across complementary observables, and is closely related to notions of quantum uncertainty and information complementarity.

The second, broader approach consists of averaging coherence over all possible orthogonal bases $\Pi$ (equivalently, over the entire unitary group with respect to the Haar measure). That is
$$
C_{\mathcal{U}}(\rho)=\int_{\mathcal{U}} C\left(\rho | U \Pi U^{\dagger}\right) d U
$$
where $d U$ denotes the normalized Haar measure on the unitary group $\mathcal{U}$, $U \Pi U^{\dagger}=\left\{U|i\rangle\langle i| U^{\dagger}\right\}$ for $\Pi=\{|i\ra\la i|,\ i=1,2,\cdots, d\}$ \cite{LUO20192869,fan2023average}. This corresponds to a fully basis-independent notion of average coherence, reflecting the typical coherence content of a quantum state when subjected to arbitrary orthogonal measurements. This measure is naturally invariant under unitary transformations and can be linked to the geometric structure of the state space and resource-theoretic averages.

Interestingly, although there are two distinct approaches to define average coherence: one based on averaging over all MUBs, and the other over all orthogonal bases. The two formulations ultimately yield the same quantitative result \cite{LUO20192869,fan2023average}
\begin{equation}\label{mubeqave}
C_{{mub}}(\rho)=C_{\mathcal{U}}(\rho)=\frac{d-(\mathbf{tr} \sqrt{\rho})^2}{d+1} .
\end{equation}

\section{Quantum average correlation based on coherence}\label{sec3}
Let us consider a bipartite state $\rho^{A B}$ of the composite system $H^A \otimes H^B$ with $\mathbf{dim} H^A=d_A$ and $\mathbf{dim} H^B=d_B$. The global information content of $\rho^{AB}$ with respect to the local measurement $\Pi=\{|i\rangle\langle i|: i=1,2, \cdots, d\}$ on $H^A$ is defined as \cite{sun2017quantum,PhysRevA.96.022130,PhysRevA.105.032436}
$$
C\left(\rho^{AB}|\Pi\otimes \mathbb{I}_{B}\right)=\sum_{i=1}^{d_A} I(\rho^{AB},|i\rangle\langle i|\otimes \mathbf{1}^{B}),
$$
which can be regarded as the measure of partial coherence in $H^A \otimes H^B$. Here $\mathbb{I}_{B}$ is the identity operator on $H^B$ and $\mathbf{1}^{B}$ is the corresponding identity matrix. Following the fact that
$$I(\rho^{AB},|i\rangle\langle i|\otimes \mathbf{1}^{B}) \geqslant I(\rho^A,|i\rangle\langle i|),$$
where $\rho^A$ is the reduced state by tracing the part $B$.
Then the correlation of $\rho^{AB}$ can be formula as \cite{Sun2022QuantifyingCoherence,PhysRevA.85.032117,sun2017quantum,PhysRevA.110.022418}
$$
Q\left(\rho^{AB}|\Pi\right)=C\left(\rho^{AB}|\Pi\otimes \mathbb{I}_{B}\right)-C\left(\rho^{A}\otimes \rho^{B}|\Pi\otimes \mathbb{I}_{B}\right).
$$
It can also be equivalently expressed as 
$$Q\left(\rho^{AB}|\Pi\right)=C\left(\rho^{AB}|\Pi\otimes \mathbb{I}_{B}\right)-C\left(\rho^{A}|\Pi\right).$$
The quantity of correlation can be interpreted as the difference between the global skew information and local skew information.

Motivated by definitions of average coherence, we also adopt two distinct approaches to define average correlation. One is based on averaging over all MUBs $\Pi_t=\left\{\left| b_{t i}\right\rangle: j=1,2, \cdots, d_A\right\}, t=1,2, \cdots, d_A+1$, 
\begin{align*}
Q_{{mub}}\left(\rho^{AB}\right)
&=\frac{1}{d_A+1} \sum_{t=1}^{d_A+1} Q\left(\rho^{AB}|\Pi_t\right)\\
&=\frac{1}{d_A+1} \sum_{t=1}^{d_A+1} \left(C\left(\rho^{AB}|\Pi_t\otimes \mathbb{I}_{B}\right)- C\left(\rho^{A}|\Pi_t\right)\right)\\
&=\frac{1}{d_A+1} \sum_{t=1}^{d_A+1}\sum_{i=1}^{d_A} I\left(\rho^{AB}, | b_{t i}\rangle \langle b_{t i}|\otimes \mathbf{1}^{B}\right)\\
&\ \quad -\frac{1}{d_A+1} \sum_{t=1}^{d_A+1}\sum_{i=1}^{d_A} I\left(\rho^{A},| b_{t i}\rangle \langle b_{t i}| \right).
\end{align*}

\begin{proposition}
The quantifier $Q_{{mub}}\left(\rho^{AB}\right)$ has the following properties:

$\mathrm{(a)}$ (Non-negativity) $Q_{{mub}}\left(\rho^{AB}\right)\geqslant 0$ and the equality holds if $\rho^{AB}=\rho^A\otimes \rho^B$ is a product state.

$\mathrm{(b)}$ (Contractivity) For any quantum channel $\Phi$ acting on party $B$,
$$
Q_{{mub}}\left((\mathcal{I}_A\otimes \Phi )(\rho^{AB})\right)\leqslant Q_{{mub}}\left(\rho^{AB}\right),
$$
where $\mathcal{I}_A$ denotes the identity channel on party $A$.

$\mathrm{(c)}$  (Local unitary invariance) Let $\sqrt{\rho^{AB}}=\sum_{uv}X_{uv}\otimes |u\ra \la v|$ with $|u\ra$ the computational basis of $H^B$. For any local unitary operators $U_A$ and $U_B$ on parties $A$ and $B$, respectively, 
$$
Q_{{mub}}\left(U_A\otimes U_B \rho^{AB}(U_A\otimes U_B)^\dagger \right)= Q_{{mub}}\left(\rho^{AB}\right),
$$
and 
\begin{equation}\label{prop1}
Q_{{mub}}\left(\rho^{AB}\right)= \frac{\left(\mathbf{tr}\sqrt{\rho^{A}}\right)^2-\mathbf{tr}_B\left(\mathbf{tr}_A \sqrt{\rho^{AB}}\right)^2}{d_A+1}.
\end{equation}
\end{proposition}
\textit{Proof} (a) Follow the fact that $I(\rho^{AB},|i\rangle\langle i|\otimes \mathbf{1}^{B}) \geqslant I(\rho^A,|i\rangle\langle i|)$ and the equality holds if $\rho^{AB}=\rho^A\otimes \rho^B$.

(b) For any quantum channel $\Phi$ acting on party $B$,
\begin{align*}
&Q\left((\mathcal{I}_A\otimes \Phi) (\rho^{AB})|\Pi_t\right)\\
&=C\left((\mathcal{I}_A\otimes \Phi) (\rho^{AB})|\Pi_t\otimes \mathbb{I}_{B}\right)\\
&\quad-C\left((\mathcal{I}_A\otimes \Phi) (\rho^{A}\otimes \rho^{B})|\Pi_t\otimes \mathbb{I}_{B}\right)\\
&=C\left((\mathcal{I}_A\otimes \Phi) (\rho^{AB})|\Pi_t\otimes \mathbb{I}_{B}\right)-C\left(\rho^{A}|\Pi_t\right)\\
&\leqslant  C\left(\rho^{AB}|\Pi_t\otimes \mathbb{I}_{B}\right)-C\left(\rho^{A}|\Pi_t\right),
\end{align*}
where the inequality is due to the contractivity of $C\left(\rho^{AB}|\Pi_t\otimes \mathbb{I}_{B}\right)$, $i.e.$ $C\left((\mathcal{I}_A\otimes \Phi) (\rho^{AB})|\Pi_t\otimes \mathbb{I}_{B}\right)\leqslant C\left(\rho^{AB}|\Pi_t\otimes \mathbb{I}_{B}\right)$\cite{sun2017quantum}. Complete the proof by the addition of all inequalities for $\Pi_t$.

(c) 
Denote $|\tilde{b}_{t i}\ra=U_A^\dagger | b_{t i}\rangle$, the new bases $\tilde{\Pi}_t=\left\{| \tilde{b}_{t i}\rangle: j=1,2, \cdots, d_A\right\}, t=1,2, \cdots, d_A+1$ are still a complete set of MUBs for any unitary operator $U_A$. In order to prove (\ref{prop1}), we need the following identities \cite{fan2023average}:
$$
\sum_{t=1}^{d_A+1}\sum_{i=1}^{d_A} |\tilde{b}_{t i}\ra \la \tilde{b}_{t i}|=(d_A+1)\mathbf{1}^{A}
$$
and
$$\sum_{t=1}^{d_A+1}\sum_{i=1}^{d_A} |\tilde{b}_{t i}\ra \la \tilde{b}_{t i}|\otimes |\tilde{b}_{t i}\ra \la \tilde{b}_{t i}|=\mathbf{1}^{A}\otimes \mathbf{1}^{A}+F,$$
where $\mathbf{1}^{A}$ is the identity matrix on $H^A$, $F=\sum_{kl}|k\ra\la l|\otimes |l\ra\la k|$ is the swap matrix and $\mathbf{tr}F(\rho\otimes \sigma)=\mathbf{tr}(\rho\sigma)$.

Let $\sqrt{\rho^{AB}}=\sum_{uv}X_{uv}\otimes |u\ra \la v|$ with $X_{uv}=\mathbf{tr}_B(\sqrt{\rho^{AB}}\mathbf{1}^{A}\otimes |v\ra\la u|)$. For any two local unitary operators $U_A$ and $U_B$, 
\begin{align*}
&\sum_{t=1}^{d_A+1}\sum_{i=1}^{d_A} I\left(U_A\otimes U_B\rho^{AB}(U_A\otimes U_B)^\dagger, | b_{t i}\rangle \langle b_{t i}|\otimes \mathbf{1}^{B}\right)\\
&=\sum_{t=1}^{d_A+1}\sum_{i=1}^{d_A} \mathbf{tr} \left(U_A\otimes U_B\rho^{AB}(U_A\otimes U_B)^\dagger | b_{t i}\rangle \langle b_{t i}|\otimes \mathbf{1}^{B} \right)\\
&-\sum_{t=1}^{d_A+1}\sum_{i=1}^{d_A} \mathbf{tr}\left(U_A\otimes U_B\sqrt{\rho^{AB}}(U_A\otimes U_B)^\dagger | b_{t i}\rangle \langle b_{t i}| \otimes \mathbf{1}^{B} \right)^2\\
&=d_A+1-\\
&\ \ \ \ \sum_{uv}\sum_{t=1}^{d_A+1}\sum_{i=1}^{d_A}\mathbf{tr}\left(U_A X_{uv} U_A^\dagger | b_{t i}\rangle \langle b_{t i}| U_A X_{vu} U_A^\dagger | b_{t i}\rangle \langle b_{t i}|\right)\\
&=d_A+1-\sum_{uv}\sum_{t=1}^{d_A+1}\sum_{i=1}^{d_A}\mathbf{tr}\left(X_{uv}|\tilde{b}_{t i}\ra \la \tilde{b}_{t i}| X_{vu}|\tilde{b}_{t i}\ra \la \tilde{b}_{t i}| \right)\\
&=d_A+1-\sum_{uv}\sum_{t=1}^{d_A+1}\sum_{i=1}^{d_A}\mathbf{tr}\left(|\tilde{b}_{t i}\ra \la \tilde{b}_{t i}|\otimes |\tilde{b}_{t i}\ra \la \tilde{b}_{t i}| X_{uv}\otimes X_{vu} \right)\\
&=d_A+1-\sum_{uv}\mathbf{tr} \left({(\mathbf{1}^{A}\otimes \mathbf{1}^{A}+F)}\left(X_{uv}\otimes X_{vu}\right)\right)\\
&=d_A+1-\sum_{uv}\mathbf{tr}\left(X_{uv}\right)\mathbf{tr}\left(X_{vu}\right)-\sum_{uv}\mathbf{tr}\left(X_{uv}X_{vu}\right)\\
&=d_A+1-\sum_{uv}\mathbf{tr}\left(X_{uv}\right)\mathbf{tr}\left(X_{vu}\right)-1\\
&=d_A-\sum_{uv}|\mathbf{tr}\left(X_{uv}\right)|^2\\
&=d_A-\mathbf{tr}_B\left(\mathbf{tr}_A \sqrt{\rho^{AB}}\right)^2
\end{align*}
and 
\begin{align*}
 &\sum_{t=1}^{d_A+1}\sum_{i=1}^{d_A} I\left(U_A\rho^{A}U_A^\dagger,| b_{t i}\rangle \langle b_{t i}| \right)\\
&=\sum_{t=1}^{d_A+1}\sum_{i=1}^{d_A} \mathbf{tr} \left(U_A\rho^{A}U_A^\dagger | b_{t i}\rangle \langle b_{t i}|\right)\\
&\quad-\sum_{t=1}^{d_A+1}\sum_{i=1}^{d_A} \mathbf{tr}\left(U_A\sqrt{\rho^{A}}U_A^\dagger | b_{t i}\rangle \langle b_{t i}| \right)^2\\
&=d_A+1-\sum_{t=1}^{d_A+1}\sum_{i=1}^{d_A}  |\la \tilde{b}_{t i}| \sqrt{\rho^{A}} |\tilde{b}_{t i}\ra|^2\\
&=d_A+1-\sum_{t=1}^{d_A+1}\sum_{i=1}^{d_A}\mathbf{tr} \left(|\tilde{b}_{t i}\ra \la \tilde{b}_{t i}|\otimes |\tilde{b}_{t i}\ra \la \tilde{b}_{t i}| \sqrt{\rho^{A}}\otimes \sqrt{\rho^{A}}\right)\\
&=d_A+1-\mathbf{tr} \left({(\mathbf{1}^{A}\otimes \mathbf{1}^{A}+F)}\left(\sqrt{\rho^{A}}\otimes \sqrt{\rho^{A}}\right)\right)\\
&=d_A-\left(\mathbf{tr}\sqrt{\rho^{A}}\right)^2.
\end{align*}
Combine the above two formulas, we obtain
$$
 Q_{{mub}}\left(\rho^{AB}\right)= \frac{\left(\mathbf{tr}\sqrt{\rho^{A}}\right)^2-\mathbf{tr}_B\left(\mathbf{tr}_A \sqrt{\rho^{AB}}\right)^2}{d_A+1}. 
$$ 
$\hfill\Box$
 
 Another way to define average correlation is averaging over all orthogonal bases
$$
Q_{\mathcal{U}}(\rho^{AB})=\int_{\mathcal{U}} Q\left(\rho^{AB} | U \Pi U^{\dagger}\right) d U.
$$
\begin{proposition}
For any state $\rho^{AB}$ in Hilbert space $H^A \otimes H^B$, the average correlation based on all orthogonal bases is
\begin{equation}\label{prop2}
Q_{\mathcal{U}}(\rho)= \frac{\left(\mathbf{tr}\sqrt{\rho^{A}}\right)^2-\mathbf{tr}_B\left(\mathbf{tr}_A \sqrt{\rho^{AB}}\right)^2}{d_A+1}.
\end{equation}
\end{proposition}
\textit{Proof} Let $\sqrt{\rho^{AB}}=\sum_{uv}X_{uv}\otimes |u\ra \la v|$. Following the relation \cite{zhang2024matrixintegralsunitarygroups}
\begin{equation}\label{zhanglinunitary}
\begin{aligned} 
\int_{\mathcal{U}} {U}^{\dagger} {A} {U} {X} {U}^{\dagger} {B} {U} dU =& \frac{d \mathbf{tr}({A} {B})-\mathbf{tr}({A}) \mathbf{tr}({B})}{d\left(d^2-1\right)} \mathbf{tr}({X}) \mathbf{1}\\
&+\frac{d \mathbf{tr}({A}) \mathbf{tr}({B})-\mathbf{tr}({A} {B})}{d\left(d^2-1\right)} {X},
\end{aligned}
\end{equation}
where $A$, $B$ and $X$ are operators with $d$ dimension, $\mathbf{1}$ is the identity matrix. 
One obtains
\begin{align*}
&\int_{\mathcal{U}}\sum_i I(\rho^{AB},U|i\ra\la i|U^\dagger\otimes \mathbf{1}^{B})d U\\
&=\sum_i \int_{\mathcal{U}}\mathbf{tr}(\rho^{AB}U|i\ra\la i|U^\dagger\otimes \mathbf{1}^{B})\\
&\quad-\mathbf{tr}(\sqrt{\rho^{AB}}U|i\ra\la i|U^\dagger\otimes \mathbf{1}^{B})^2dU\\
&=1-\sum_i\sum_{uv}\mathbf{tr}(|i\ra\la i|\int_{\mathcal{U}}U^\dagger X_{vu} U|i\ra\la i|U^\dagger X_{uv}UdU )\\
&=1-\sum_i\sum_{uv} \mathbf{tr}(|i\ra\la i|\frac{d_A\mathbf{tr}(X_{uv}X_{vu})-\mathbf{tr}(X_{uv})\mathbf{tr}(X_{vu})}{d_A (d_A^2-1)})\\
&\quad-\sum_i\sum_{uv} \mathbf{tr}(|i\ra\la i|\frac{d_A\mathbf{tr}(X_{uv})\mathbf{tr}(X_{vu})-\mathbf{tr}(X_{uv}X_{vu})}{d_A (d_A^2-1)})\\
&=1-\sum_{uv}\frac{\mathbf{tr}(X_{uv})\mathbf{tr}(X_{vu})-\mathbf{tr}(X_{uv}X_{vu})}{d_A+1}\\
&=1-\frac{\sum_{uv}|\mathbf{tr}\left(X_{uv}\right)|^2-1}{d_A+1}\\
&=\frac{d_A-\sum_{uv}|\mathbf{tr}\left(X_{uv}\right)|^2}{d_A+1}\\
&=\frac{d_A-\mathbf{tr}_B\left(\mathbf{tr}_A \sqrt{\rho^{AB}}\right)^2}{d_A+1}.
\end{align*}
Complete the proof by combing with the equality (\ref{mubeqave}).$\hfill\Box$

In Ref.~\cite{PhysRevA.85.032117}, Luo et al. proposed a well-defined measure of correlation based on local observables, 
\begin{equation}\label{luocorr2012}
Q_{ob}(\rho^{AB})=\frac{1}{d_A+1}\sum_{i=1}^{d_A^2} I(\rho^{AB},G_i\otimes \mathbf{1}^{B}) - I(\rho^A,G_i),
\end{equation}
where $\{G_i\}_{i=1}^{d_A^2}$ is an orthonormal base for the real Hilbert space $\mathcal{L}
(H^A)$. Note that the quantum correlation (\ref{luocorr2012}) is basis-independent. Without loss of generality, let $\{|i\rangle\}$ be an orthonormal base for $H^A$, $\mathbf{i}$ be imaginary unit, and
$$
\begin{aligned}
S_{i j} & :=\frac{1}{\sqrt{2}}(|i\rangle\langle j|+|j\rangle\langle i|), \\
T_{i j} & :=\frac{\mathbf{i}}{\sqrt{2}}(|i\rangle\langle j|-|j\rangle\langle i|),
\end{aligned}
$$
then
$$
\{|i\rangle\langle i|\} \cup\left\{S_{i j}: i<j\right\} \cup\left\{T_{i j}: i<j\right\}
$$
constitutes an orthonormal base $\{G_i\}_{i=1}^{d_A^2}$ for $\mathcal{L}\left(H^A\right)$. Employing the relation (\ref{zhanglinunitary}) and the facts $\sum_i\mathbf{tr}(G_i)^2=d_A$ and $\sum_i\mathbf{tr}(G_i^2)=d_A^2$, for $\sqrt{\rho^{AB}}=\sum_{uv}X_{uv}\otimes |u\ra \la v|$, one obtains
\begin{align*}
&\int_{\mathcal{U}}\sum_i I(\rho^{AB},UG_iU^\dagger\otimes \mathbf{1}^{B})d U\\
&=\sum_i \int_{\mathcal{U}}\mathbf{tr}(\rho^{AB}UG_i^2U^\dagger\otimes \mathbf{1}^{B})\\
&\quad-\mathbf{tr}(\sqrt{\rho^{AB}}UG_iU^\dagger\otimes \mathbf{1}^{B})^2dU\\
&=d_A-\sum_i\sum_{uv}\mathbf{tr}(G_i\int_{\mathcal{U}}U^\dagger X_{vu} UG_iU^\dagger X_{uv}UdU )\\
&=d_A-\sum_i\sum_{uv} \frac{d_A\mathbf{tr}(X_{uv}X_{vu})-\mathbf{tr}(X_{uv})\mathbf{tr}(X_{vu})}{d_A (d_A^2-1)}(\mathbf{tr}(G_i))^2\\
&\quad -\sum_i\sum_{uv} \frac{d_A\mathbf{tr}(X_{uv})\mathbf{tr}(X_{vu})-\mathbf{tr}(X_{uv}X_{vu})}{d_A (d_A^2-1)}\mathbf{tr}(G_i^2)\\
&=d_A-\sum_{uv} \frac{d_A\mathbf{tr}(X_{uv}X_{vu})-\mathbf{tr}(X_{uv})\mathbf{tr}(X_{vu})}{(d_A^2-1)}\\
&\quad-\sum_{uv} \frac{d_A^2\mathbf{tr}(X_{uv})\mathbf{tr}(X_{vu})-d_A\mathbf{tr}(X_{uv}X_{vu})}{(d_A^2-1)}\\
&=d_A-\mathbf{tr}_B\left(\mathbf{tr}_A \sqrt{\rho^{AB}}\right)^2.
\end{align*}
The local skew information 
$$\sum_{i=1}^{d_A^2} I(\rho^A,G_i)=d_A-\left(\mathbf{tr}\sqrt{\rho^A}\right)^2.$$
Thus
\begin{equation}\label{corrob}
Q_{ob}\left(\rho^{AB}\right)=\frac{\left(\mathbf{tr}\sqrt{\rho^{A}}\right)^2-\mathbf{tr}_B\left(\mathbf{tr}_A \sqrt{\rho^{AB}}\right)^2}{d_A+1}.
\end{equation}

By inspecting the equalities (\ref{prop1}), (\ref{prop2}) and (\ref{corrob}) we have the following result.
\begin{proposition}\label{prop3}
For any state $\rho^{AB}$ in Hilbert space $H^A \otimes H^B$, the average correlation based on coherence is equivalent to the correlation based on orthonormal base, $i.e.$,
\begin{equation}
Q_{{mub}}\left(\rho^{AB}\right)=Q_{\mathcal{U}}(\rho^{AB})=Q_{ob}\left(\rho^{AB}\right).
\end{equation}
\end{proposition}

This equivalence may be explained from a group-theoretic and invariant-measure perspective. The existence of a complete set of MUBs implies that such a set provides a discrete uniform tiling of the unitary group under the Haar measure. As a result, averaging correlation over the complete MUB ensemble, integrating over the entire unitary group, or summing over any complete orthonormal operator basis ultimately yields the same value, since all three procedures correspond to the same invariant integration over the space of measurement bases under appropriate normalization. This demonstrates that the concept of average correlation is well-defined in a basis-independent manner, reflecting a deeper structural symmetry in the space of quantum states.

Actually, from the perspective of quantum channels, the average correlation can also be interpreted as the correlation relative to certain quantum channels.  The quanutm coherence based on the depolarizing channel $\mathcal{E}_{\mathrm{De}}$ can be expressed as \cite{PhysRevA.105.032436,fan2023average}:
$$
C(\rho^{A} |\mathcal{E}_{\mathrm{De}})=\sum_{i=1}^{d_A^2} I\left(\rho^A, \frac{X_i}{\sqrt{d_A+1}}\right)+I(\rho^A,\frac{\mathbf{1}^A}{\sqrt{d_A+1}}),
$$
where 
$$
\left\{\frac{1}{\sqrt{d_A+1}}\mathbf{1}^A, \frac{1}{\sqrt{d_A+1}}X_i: i=1,2,\cdots,d_A^2\right\}
$$
is the set of Kraus operators of $\mathcal{E}_{\mathrm{De}}$ and $\{X_i: i=1,2,\cdots,d_A^2\}$ is a complete orthonormal operator basis of $\mathcal{L}(H^A)$. In this sense, it is natrual to explain the averaged correlations as the correlations relative to the quantum channel $\mathcal{E}_{\mathrm{De}}$ as defined in Ref.~\cite{PhysRevA.105.032436}:
$$
\begin{aligned}
Q(\rho^{AB} |\mathcal{E}_{\mathrm{De}})&=C(\rho^{AB} |\mathcal{E}_{\mathrm{De}}\otimes \mathcal{I}_{B})-C(\rho^{A} |\mathcal{E}_{\mathrm{De}}),
\end{aligned}
$$
where $\mathcal{I}_{B}$ is the identity channel on system $B$. More generally, the correlations relative to the twirling channel $\mathcal{T}_{\mathcal{U}}$ induced by the unitary group $\mathcal{U}(H^{A})$ are actually the correlations relative to the completely depolarizing channel $\mathcal{E}_{\mathrm{De}}$. The corresponding quantifier of correlations in $\rho^{A B}$ relative to $\mathcal{T}_{\mathcal{U}}$ is defined as \cite{PhysRevA.105.032436}

$$
\begin{aligned}
 Q\left(\rho^{A B}|\mathcal{T}_{\mathcal{U}}\right)&=\int_{\mathcal{U_A}}I\left(\rho^{AB}, U_A \otimes \mathbf{1}^B\right)-I\left(\rho^A, U_A\right) d U_A\\
 &=\frac{\left(\mathbf{tr}\sqrt{\rho^{A}}\right)^2-\mathbf{tr}_B\left(\mathbf{tr}_A \sqrt{\rho^{AB}}\right)^2}{d_A}.
\end{aligned}
$$

According to the Proposition \ref{prop3}, one has
$$
Q_{ob}\left(\rho^{AB}\right)=Q(\rho^{AB} |\mathcal{E}_{\mathrm{De}})= \frac{d_A}{d_A+1}Q\left(\rho^{A B}|\mathcal{T}_{\mathcal{U}}\right).
$$

\section{Complementarity between average correlation and wave-particle duality}\label{sec3.2}

The average correlation defined in the previous section can be employed to establish a novel complementarity relation that connects wave-particle duality and system-environment correlations. Let $\rho^{A E}=|\Psi\rangle\langle\Psi|$ be a bipartite pure state of $H^A \otimes H^E$ with the Schmidt decomposition $|\Psi\rangle= \sum_{\mu=1}^{d_A} \sqrt{\lambda_\mu}|\mu\rangle \otimes\left|b_\mu\right\rangle$, then
\begin{equation}\label{avecorrpure}
Q_{\mathcal{U}}(\rho^{AE})=\frac{\left(\sum_\mu \sqrt{\lambda_\mu}\right)^2-\sum_\mu \lambda_\mu^2}{d_A+1} .
\end{equation}

In Ref.~\cite{Fu_2022}, Fu and Luo illustrated the wave-particle duality by casting it into a form of information conservation in a multi-path interferometer, with uncertainty as a unified theme. The wave feature via state uncertainty has the following form in paths $\Pi=\{|i\ra\la i|,\ i=1,2,\cdots, d_A\}$,
\begin{equation}
W(\rho^A|\Pi)=\sum_{i \neq j}|\langle i| \rho^A| j\rangle |^ 2 .
\end{equation}
The particle feature is defined as
\begin{equation}
P(\rho^A|\Pi)=\sum_{i}|\langle i| \rho^A| i\rangle |^ 2 .
\end{equation}
Thus
\begin{equation}\label{wpd}
 W(\rho^A|\Pi)+P(\rho^A|\Pi)=\mathbf{tr}\left(\rho^A\right)^2.
\end{equation}

From the perspective of quantum complementarity, the following result can be directly derived by combining equations (\ref{avecorrpure}) and (\ref{wpd}).

\begin{proposition}\label{prop4}
Let $\rho^{A E}=|\Psi\rangle\langle\Psi|$ be a bipartite state shared between quantum system A and environment E, $i.e.$, $|\Psi\rangle\in H^A\otimes H^E$. The following complementarity relation holds:
\begin{equation}\label{wpc}
W(\rho^A|\Pi)+P(\rho^A|\Pi)+(d_A+1)Q_{\mathcal{U}}(\rho^{AE})=\left(\mathbf{tr}\sqrt{\rho^A}\right)^2,
\end{equation}
where $\Pi=\{|i\ra\la i|,\ i=1,2,\cdots, d_A\}$ denotes the set of path projectors.
\end{proposition}

Based on the principle of quantum complementarity, the relation presented in Proposition \ref{prop4} can be analyzed as follows: This complementary relationship provides a unified description of wave-particle duality and the measure of system-environment correlations. Equation (\ref{wpc}) shows that for any bipartite pure state $\rho^{A E}$, the sum of the wave feature $W\left(\rho^A | \Pi\right)$ and the particle feature $P\left(\rho^A |\Pi\right)$ of the local subsystem $A$ is not conserved. The missing part is precisely compensated by the correlation measure $Q_{\mathcal{U}}\left(\rho^{A E}\right)$ between the system and the environment. This structure reveals the deeper implication of quantum complementarity in a multipath interferometer: when the correlation between the system and the environment strengthens, the local subsystem exhibits less pronounced wave-particle behavior; conversely, when the correlation weakens, the wave-particle behavior becomes more dominant. In particular, when $Q_{\mathcal{U}}\left(\rho^{A E}\right)=0$, $i.e.$, the system and environment are uncorrelated, the relation reduces to the standard wave-particle duality identity given in Eq.~(\ref{wpd}), indicating that the wave and particle aspects of subsystem $A$ then form a complete complementary pair. This proposition not only offers a quantitative tool for understanding how quantum correlations modulate complementarity, but also establishes a direct link between wave-particle duality in multi-path interference and the theory of quantum entanglement.

\section{Conclusion}\label{sec4}
In this work, we have established a coherent framework for quantifying average coherence and average correlation using the Wigner-Yanase skew information. The equivalence between the mutually unbiased bases averaged and unitarily averaged formulations demonstrates that both quantities can be defined in a basis independent way, rooted in the invariant Haar measure on the unitary group. The proposed average correlation measure fulfills fundamental physical requirements such as non negativity, contractivity, and local unitary invariance, positioning it as a well behaved descriptor of quantum correlations beyond entanglement.

Furthermore, by relating the average correlation to wave-particle duality measures, we have obtained a new complementarity relation that unifies coherence, path information, and system environment correlation. This relation underscores the inherent trade-offs among different informational aspects of a quantum state and provides a fresh viewpoint on the role of correlation in quantum interference.

Future work may extend this approach to multipartite systems, explore operational interpretations of the average correlation in quantum protocols, and examine its dynamics in open quantum systems. These results reinforce that skew information based quantities offer a natural and mathematically rigorous framework for characterizing quantum coherence, correlations, and complementarity.

\bigskip
\noindent{\bf Acknowledgments}\, \,
This work is supported by the National Natural Science Foundation of China (NSFC) (Grant Nos.~12526648, 12401397), Natural Science Foundation of Hunan Province (Grant Nos.~2025JJ60025, 2026JJ60117), Scientific Research Project of the Education Department of Hunan Province (Grant No.~24B0298) and Changsha University of Science and Technology (Grant No.~097000303923).

\bigskip
\noindent{\bf Data Availability Declarations}\, \,
No datasets were generated or analyzed during the current study.

\bigskip
\noindent{\bf Conflict of interest}\, \,
The authors declare no Conflict of interest.

\bibliography{zhang}
\bibliographystyle{iopart-num}

\end{document}